\documentclass[12pt]{iopart}

\usepackage[english]{babel}
\usepackage[utf8]{inputenc}

\usepackage{iopams}
\usepackage{amssymb,mathdots,bbm,amsopn}
\newcommand{\eqref}[1]{\eref{#1}}

\usepackage{amsthm}
\theoremstyle{plain}
\newtheorem{thm}{Theorem}
\newtheorem{lemma}[thm]{Lemma}
\newtheorem{cor}[thm]{Corollary}

\usepackage[hyphens]{url}
\usepackage[pdftex,bookmarks,colorlinks]{hyperref}

\usepackage[square,numbers,comma,sort&compress]{natbib}

\usepackage[inline]{enumitem}
\usepackage{graphicx}

\newcommand{\id}[0]{~ \rmd}
\newcommand{\ee}[0]{\rme}
\newcommand{\ii}[0]{\rmi}

\newcommand{\Hilbert}{\mathcal H}
\newcommand{\Operators}{\mathcal S}
\newcommand{\ket}[1]{| #1 \rangle}
\newcommand{\ketbra}[2]{| #1 \rangle\!\langle #2 |}

\newcommand{\abs}[1]{\vert#1\vert}

\newcommand{\norm}[1]{\Vert#1\Vert}
\newcommand{\Norm}[1]{\left\Vert#1\right\Vert}

\newcommand{\sk}[2]{\langle #1, #2 \rangle}

\newcommand{\comm}[2]{[ #1, #2 ]}

\newcommand{\trbr}[1]{\tr[#1]}

\newcommand{\lspan}[1]{\operatorname{span}\{ #1 \}}

\renewcommand{\vec}{\mathbf}

\begin{document}

\title{Limit cycles in periodically driven open quantum systems}

\author{Paul Menczel$^1$ and Kay Brandner$^1$}

\address{$^1$ Department of Applied Physics, Aalto University, 00076 Aalto, Finland}
\ead{paul@menczel.net}

\begin{abstract}
We investigate the long-time behavior of quantum $N$-level systems that are coupled to a Markovian environment and subject to periodic driving.
As our main result, we obtain a general algebraic condition ensuring that all solutions of a periodic quantum master equation with Lindblad form approach a unique limit cycle. Quite intuitively, this criterion requires that the dissipative terms of the master equation connect all subspaces of the system Hilbert space during an arbitrarily small fraction of the cycle time.
Our results provide a natural extension of Spohn's \emph{algebraic condition for the approach to equilibrium} to systems with external driving.
Moreover, our theory leads to a rigorous condition for the emergence of dissipative discrete time crystals and covers also classical, periodically modulated Markov jump processes.
\end{abstract}

\renewcommand{\abstractname}{Corrigendum}
\begin{abstract}
This paper, published in 2019, contains an essential mathematical error that invalidates its derivations except in the special case of unital GKSL generators.
For details, see the Corrigendum appended to this document in 2026.
\end{abstract}
\vspace{1em}

\section{Introduction and Main Result}

The theory of open quantum systems provides us with powerful tools to describe the dynamics of quantum-mechanical objects that interact with a macroscopic environment \cite{Breuer2002}.
A cornerstone of this framework is the Gorini–Kossakowski–Sudarshan–Lindblad (GKSL) equation,
\begin{equation} \label{eq:lindblad}
	\partial_t \rho_t = \hat L_t \rho_t \equiv \frac{1}{\ii\hbar} \comm{H_t}{\rho_t}
		+ \sum_{\mu=1}^{M_t} \gamma^\mu_t \Big( A^\mu_t \rho_t A^{\mu\dagger}_t - \frac 1 2 A^{\mu\dagger}_t A^\mu_t \rho_t - \frac 1 2 \rho_t A^{\mu\dagger}_t A^\mu_t \Big) .
\end{equation}
Under certain technical conditions, this equation of motion defines the most general Markovian, i.e., memoryless, time evolution of a physical state $\rho_t$ \cite{GoriniJMathPhys1976,LindbladCommunMathPhys1976,ChruscinskiJPhysB2012,Rivas2012}.
The self-adjoint operator $H_t$ is thereby usually identified with the effective Hamiltonian of the system, while the influence of the environment is encoded in the Lindblad operators $A^\mu_t$ with corresponding coupling rates $\gamma^\mu_t > 0$; $\hbar$ denotes Planck's constant.
Notably, all components of the Lindblad generator $\hat L_t$, including the number of dissipation channels $M_t$, can be time-dependent if the system is externally driven.

Owing to its high degree of generality, the GKSL equation has found application in nearly all areas of modern quantum physics ranging from quantum optics \cite{Scully1997} and quantum information theory \cite{Nielsen2000,GooldJPhysA2016} over quantum thermodynamics \cite{KosloffEntropy2013,VinjanampathyContempPhys2016,BenentiPhysRep2017} to quantum device engineering \cite{PekolaNatPhys2015}.
The general mathematical properties of this equation, which have been extensively studied over the last decades \cite{GoriniRepMathPhys1978,SpohnJMathPhys1978,AlickiCommunMathPhys1987,Alicki2007,SchirmerPhysRevA2010,AlbertPhysRevX2016,BreuerRevModPhys2016}, have thus become a source of valuable physical insight.
A particularly important problem in this context is to characterize the long-time behavior of its solutions.
For systems without external driving, which are described by a time-independent Lindblad generator, the conditions that lead to a unique stationary limiting state $\rho^{\mathrm{ss}}$ are well understood \cite{Rivas2012}.
In particular, Spohn proved a general criterion that, quite intuitively, requires the dissipative terms of the Lindblad generator to connect all subspaces of the system Hilbert space \cite{SpohnLettMathPhys1977}.
This result can be formulated as follows:

\begin{thm}[Spohn] \label{thm:spohn}
	For a system obeying the GKSL equation \eqref{eq:lindblad} with time-independent Lindblad generator $\hat L$, let $\mathcal A \equiv \lspan{ A^1 \dots A^M }$ be the linear span of all Lindblad operators.
	If $\mathcal A$ is self-adjoint and irreducible, then there exists a unique state $\rho^{\mathrm{ss}}$ so that
	\begin{equation}
		\lim_{t \to \infty} \left( \rho_t - \rho^{\mathrm{ss}} \right) = 0
	\end{equation}
	for arbitrary initial conditions $\rho_0$.
	Here, $\mathcal A$ is self-adjoint if $X^\dagger \in \mathcal A$ for all $X \in \mathcal A$ and irreducible if $\comm X Y = 0$ for all $X \in \mathcal A$ implies that $Y$ is a multiple of the identity.
\end{thm}

The central aim of this letter is to extend this theorem to periodically driven systems, whose Lindblad generator is explicitly time-dependent and obeys the condition $\hat L_{t+T} = \hat L_t$ for some cycle time $T$.
Such systems are commonly studied in quantum thermodynamics, for example as models of cyclic thermal machines \cite{KosloffEntropy2013,BrandnerPhysRevE2016,VinjanampathyContempPhys2016}.

When a periodic driving protocol is constantly applied to a dissipative system, we intuitively expect its state $\rho_t$ to approach a periodic limit cycle satisfying $\rho^{\mathrm{cy}}_{t+T} = \rho^{\mathrm{cy}}_t$.
This expectation can be motivated using the following argument given in \cite{KosloffEntropy2013}.
The relative entropy between two states $\rho_t$ and $\sigma_t$ is defined as
\begin{equation} \label{eq:relentro}
	S(\rho_t \Vert \sigma_t) = -\trbr{\rho_t\, (\log \rho_t - \log \sigma_t)} .
\end{equation}
Being non-negative and zero only if $\rho_t = \sigma_t$, this quantity can be understood as a distance measure on the state space of a given system.
The dynamics induced by the GKSL equation can only decrease the relative entropy between two states, i.e., $S(\rho_\tau \Vert \sigma_\tau) \leq S(\rho_t \Vert \sigma_t)$ for any time $\tau > t$ \cite{Breuer2002}.
Using this result with $\sigma_t = \rho_{t+T}$ and $\tau = t+T$ yields $S(\rho_{t+T} \Vert \rho_{t+2T}) \leq S(\rho_{t} \Vert \rho_{t+T})$.
Hence, the relative entropy between $\rho_t$ and $\rho_{t+T}$ decreases with every period.

This argument implies that the relative entropy between $\rho_t$ and $\rho_{t+T}$ converges to a limit value at long times.
However, this value does not have to be zero, i.e., $\rho_t$ might not be equal to $\rho_{t+T}$.
In fact, recent studies show that systems that do not adopt the periodicity of an applied driving field, generally know as discrete time crystals~\cite{SachaRepProgPhys2017}, can indeed be implemented with GKSL dynamics~\cite{GongPhysRevLett2018, GambettaPhysRevLett2019}.
The period length of the system response at long times is thereby an integer multiple of the period $T$ of the applied driving.
Hence, these results present examples for situations, where the relative entropy between $\rho_t$ and $\rho_{t+T}$ does not converge to zero.

Furthermore, the bare fact that the relative entropy \eqref{eq:relentro} cannot increase over time does not lead to any conclusions on the uniqueness of a potential limit cycle.
It thus remains unclear which conditions must be met for the long-time solution of \eqref{eq:lindblad} to be a unique limit cycle satisfying $\rho^{\mathrm{cy}}_{t+T} = \rho^{\mathrm{cy}}_t$.
Here, we show that, to this end, it suffices that the periodic Lindblad generator $\hat L_t$ satisfies the requirements of Theorem~\ref{thm:spohn} for an arbitrarily small fraction of the period $T$.
As our main result, we obtain the following theorem:

\begin{thm} \label{thm:us}
	Consider a system obeying the GKSL equation \eqref{eq:lindblad} with periodic Lindblad generator, $\hat L_{t+T} = \hat L_t$.
	Assume there is a $\tau \in (0, T]$ so that $\hat L_t$ is continuous on $[0,\tau]$, and $\mathcal A_t \equiv \lspan{ A^1_t \dots A^{M_t}_t }$ is self-adjoint and irreducible for all $t \in [0, \tau]$.
	Then there exists a unique limit cycle $\rho^{\mathrm{cy}}_t = \rho^{\mathrm{cy}}_{t+T}$ so that
	\begin{equation}
		\lim_{t \to \infty} \left( \rho_t - \rho^{\mathrm{cy}}_t \right) = 0
	\end{equation}
	for arbitrary initial conditions $\rho_0$.
\end{thm}

We proceed as follows.
After fixing the notation in \Sref{sec:notation}, we briefly review the proof of Theorem~\ref{thm:spohn} in \Sref{sec:no_driving}.
In \Sref{sec:driving}, we derive our main result.
Further aspects are considered in our concluding \Sref{sec:conclusions}.

\section{Setup and Notation} \label{sec:notation}

We consider an open quantum system with finite-dimensional Hilbert space $\Hilbert$.
The space of linear, self-adjoint operators $X: \Hilbert \to \Hilbert$ is denoted by $\Operators$ and equipped with the trace-norm
\begin{equation}
	\norm{X}_1 \equiv \sum_{\lambda \in \sigma(X)} \abs{\lambda} ,
\end{equation}
where $\sigma(X)$ is the set of eigenvalues of $X$.
Upon introducing the Hilbert-Schmidt inner product,
\begin{equation} \label{eq:hs}
	\sk X Y \equiv \trbr{ X Y } ,
\end{equation}
$\Operators$ becomes a Hilbert space \cite{Reed1981}.
Linear operators $\hat V: \Operators \to \Operators$ are called superoperators and indicated by hats; $\hat V^\dagger$ denotes the Hermitian adjoint of $\hat V$ with respect to the scalar product \eqref{eq:hs}.
The operator norm of a superoperator with respect to the trace-norm is defined as
\begin{equation}
	\norm{\hat V}_1 \equiv \max_{X \in \Operators,\, X \neq 0} \frac{\norm{\hat V X}_1}{\norm{X}_1} .
\end{equation}
Throughout this letter, we use the symbol $\rho$ to denote the state of the system, i.e., a positive operator in $\Operators$ satisfying $\tr \rho = 1$.

\section{Time-Independent Lindblad Generators} \label{sec:no_driving}

In this section, we sketch Spohn's proof of Theorem~\ref{thm:spohn} \cite{SpohnLettMathPhys1977}.
We consider an open quantum system obeying the GKSL equation \eqref{eq:lindblad} with constant generator $\hat L_t = \hat L$.
The formal solution of this differential equation is given by $\rho_t = \hat V_t\, \rho_0$, where the set of propagators $\hat V_t \equiv \ee^{\hat L t}$ forms a quantum dynamical semigroup.
To prove Theorem~\ref{thm:spohn}, we will show that $\hat L$ has exactly one zero eigenvalue and that all other eigenvalues of $\hat L$ have negative real parts.

Recall that $\mathcal A \equiv \lspan{ A^1 \dots A^M }$ denotes the linear span of the Lindblad operators.
Since $\mathcal A$ is self-adjoint by assumption, there is an orthonormal set of Hermitian operators $\{ F^1 \dots F^m \} \subset \Operators$ with $m \leq M$ that spans $\mathcal A$, i.e., we have $A^\mu \equiv \sum_{\alpha=1}^m c_{\mu\alpha} F^\alpha$ for some complex coefficients $c_{\mu\alpha}$.
Upon inserting this expansion into \eqref{eq:lindblad}, the Lindblad generator takes the form
\begin{equation} \label{eq:lindblad_expanded}
	\hat L \rho = \frac{1}{\ii\hbar} \comm{H}{\rho} + \sum_{\alpha,\beta = 1}^m B_{\alpha\beta}\, \Big( F^\alpha \rho F^\beta - \frac 1 2 F^\beta F^\alpha \rho - \frac 1 2 \rho F^\beta F^\alpha \Big) ,
\end{equation}
where $B_{\alpha\beta} \equiv \sum_{\mu=1}^M \gamma^\mu c_{\mu\alpha} c_{\mu\beta}^\ast$.
By construction, the coefficient matrix $B_{\alpha\beta}$ is positive definite.
Hence, its smallest eigenvalue $b$ is strictly positive.
We now separate the diagonal contribution
\begin{equation} \label{eq:Ls}
	\hat L_{\mathrm{d}}\, \rho \equiv \frac b 2\, \sum_{\alpha=1}^m  \Big( F^\alpha \rho F^\alpha - \frac 1 2 F^\alpha F^\alpha \rho - \frac 1 2 \rho F^\alpha F^\alpha \Big)
\end{equation}
from the generator \eqref{eq:lindblad_expanded}, such that $\hat L \equiv \hat L_{\mathrm{d}} + \hat L_{\mathrm{r}}$.

Using the irreducibility of $\mathcal A$, it is straightforward to show that the superoperator $\hat L_{\mathrm{d}}$ has exactly one zero eigenvalue corresponding to the eigenvector $\mathbbm 1$, and that its remaining eigenvalues are negative \cite{SpohnLettMathPhys1977}.
We now denote by $-\lambda$ the maximum non-vanishing eigenvalue of $\hat L_{\mathrm{d}}$ and by $\hat L^\prime$, $\hat L_{\mathrm{d}}^\prime$ and $\hat L_{\mathrm{r}}^\prime$ the restrictions of $\hat L$, $\hat L_{\mathrm{d}}$ and $\hat L_{\mathrm{r}}$ to $\mathbbm 1^\perp$, the orthogonal complement of $\mathbbm 1$ in $\Operators$, i.e., the subspace of traceless operators.
Since $\hat L_{\mathrm{d}}$ is self-adjoint, we can conclude that
\begin{equation} \label{eq:LdPrime}
	\norm{ \exp[ \hat L_{\mathrm{d}}^\prime t ]\, }_1^\prime \leq \ee^{-\lambda t} .
\end{equation}
Here, $\norm{\bullet}_1^\prime$ is the operator norm with respect to the trace-norm on $\mathbbm 1^\perp$.

Next, we consider the remaining generator $\hat L_{\mathrm{r}}$.
By construction, this superoperator is still the generator of a quantum dynamical semigroup.
Hence, the corresponding propagator $\exp[ \hat L_{\mathrm{r}} t ]$ is completely positive and trace preserving and therefore con\-trac\-tive, i.e., $\norm{ \exp[ \hat L_{\mathrm{r}} t ]\, }_1 = 1$ \cite{Rivas2012}.
Since the image of $\hat L_{\mathrm{r}}$ is contained in $\mathbbm 1^\perp$, we find that
\begin{equation}
	\norm{ \exp[ \hat L_{\mathrm{r}}^\prime t ]\, }_1^\prime \leq \norm{ \exp[ \hat L_{\mathrm{r}} t ]\, }_1 = 1 .
\end{equation}
Combining this bound with \eqref{eq:LdPrime} by means of the Lie product formula \cite{Reed1981}, we obtain
\begin{equation} \label{eq:1norm}
	\norm{ \exp[ \hat L' t ]\, }_1^\prime \leq \lim_{N \to \infty} \Big( \Norm{ \, \exp[\hat L_{\mathrm{d}}^\prime\, t/N]\, \exp[\hat L_{\mathrm{r}}^\prime\, t/N]\, }_1' \Big)^N \leq \ee^{-\lambda t} .
\end{equation}
Thus, all eigenvalues of $\hat L^\prime$ have negative real parts.

Finally, we choose an orthonormal basis of $\Operators$ with $(\dim \Hilbert)^{-1/2}\, \mathbbm 1$ being the first basis element.
The matrix representation of $\hat L$ in such a basis has the form
\begin{equation}
	\hat L = \Bigg[\! \begin{array}{cc}
			0 & 0 \cdots 0 \\
			\ast & (\hat L^\prime)
		\end{array} \!\Bigg] ,
\end{equation}
where $\ast$ denotes unknown entries. 
Examining the characteristic polynomial of this matrix proves first that $\hat L$ has a single zero eigenvalue and second that the remaining eigenvalues are identical to the eigenvalues of $\hat L^\prime$, which have negative real parts.
The proof of Theorem~\ref{thm:spohn} is thus complete.

\section{Periodic Lindblad Generators} \label{sec:driving}

We now move on to driven systems with a time-periodic Lindblad generator $\hat L_{t+T} = \hat L_t$.
The propagator $\hat V_{t,t_0}$, which maps the initial state $\rho_{t_0}$ to the later state $\rho_t$, here depends on both the initial and the final time.
Our goal is to show that $\rho_t$ approaches a limit cycle $\rho^{\mathrm{cy}}_{t+T} = \rho^{\mathrm{cy}}_t$ at long times.
A direct approach to this problem is, however, complicated by the fact that the limit cycle can generally not be explicitly determined.

To circumvent this difficulty, we work in the Heisenberg picture, where the state $\rho_0$ is constant and the observables $X \in \Operators$ carry the time dependence $X_t = V_{t,0}^\dagger X_0$.
Note that the time evolution of $X_t$ does generally not follow from a time-local differential equation \cite{Breuer2002}.
In order to prove Theorem~\ref{thm:us}, it suffices to show that any observable $X_t$ becomes a multiple of the identity at long times, $X_t = x_t \mathbbm 1$ for some scalar $x_t = x_{t+T}$; the expectation value $\langle X_t \rangle \equiv \trbr{X_t\, \rho_0}$ then becomes periodic and independent of the initial state.

The eigenvalues of $\hat L_t^\dagger$ are the complex conjugates of the eigenvalues of $\hat L_t$.
Therefore, they all have a negative real part, except for a single zero eigenvalue with constant eigenvector $\mathbbm 1$. 
This fact is, however, not sufficient to conclude that $X_t$ becomes a multiple of the identity at long times, since the adjoint generators $\hat L_{t_1}^\dagger$ and $\hat L_{t_2}^\dagger$ at different times $t_1$ and $t_2$ do not commute with each other in general.
Instead, the strategy of our proof is to introduce a norm for $X_t'$, the orthogonal projection of $X_t$ on $\mathbbm 1^\perp$, which strictly decreases whenever $\hat L_t$ satisfies the conditions of Theorem~\ref{thm:spohn}.

To this end, we first consider a general subspace $\Operators'$ of $\Operators$, and define
\begin{equation} \label{eq:norm}
	\norm{X}_\infty \equiv \max \left\{ \abs{ \sk Y X } : \norm{Y}_1 = 1, Y \in \Operators' \right\}
\end{equation}
for $X \in \Operators'$.
It is easy to show that \eqref{eq:norm} indeed defines a norm on $\Operators'$.
Note that \eqref{eq:norm} coincides with the usual definition of the operator norm $\norm{\bullet}_\infty$ in the case $\Operators' = \Operators$.

\begin{lemma} \label{lemma}
	Let $\hat U: \Operators' \to \Operators'$ be a superoperator on the subspace $\Operators'$.
	Then $\norm{\hat U}_\infty$, the operator norm of $\hat U$ with respect to $\norm{\bullet}_\infty$, satisfies
	\begin{equation}
		\norm{\hat U}_\infty \leq \norm{\hat U^\dagger}_1' ,
	\end{equation}
	where $\norm{\bullet}_1'$ is the operator norm with respect to the trace-norm on $\Operators'$.
\end{lemma}
\begin{proof}
	For $Y \in \Operators'$ with $\norm{Y}_1 = 1$, we first define $N(Y) \equiv \hat U^\dagger Y /\, \norm{\hat U^\dagger Y}_1$.
	If $\hat U^\dagger Y = 0$, we set $N(Y) = Y$ such that $\hat U^\dagger Y = \norm{\hat U^\dagger Y}_1\, N(Y)$ still holds.
	Note that $N(Y) \in \Operators'$ with $\norm{N(Y)}_1 = 1$.
	Let $X \in \Operators'$, then
	\begin{eqnarray}
		\norm{\hat U X}_\infty &= \max \left\{ \abs{ \sk Y {\hat U X} } : \norm{Y}_1 = 1, Y \in \Operators' \right\} \nonumber \\
			&= \max \left\{ \norm{\hat U^\dagger Y}_1\, \abs{ \sk {N(Y)} X } : \norm{Y}_1 = 1, Y \in \Operators' \right\} \nonumber \\
			&\leq \norm{\hat U^\dagger}_1'\, \norm{X}_\infty . 
	\end{eqnarray}
	It follows that $\norm{\hat U}_\infty \leq \norm{\hat U^\dagger}_1'$.
\end{proof}

We are now ready to prove Theorem~\ref{thm:us}.
We assume that the generator $\hat L_t$ satisfies the conditions of Theorem~\ref{thm:spohn} for $0 \leq t \leq \tau$.
The adjoint propagator over this interval is given by the ordered exponential
\begin{equation} \label{eq:toe}
	V_{\tau,0}^\dagger = \lim_{N \to \infty} \left( \exp[ \hat L_{\Delta t}\, \Delta t ]^\dagger\; \exp[ \hat L_{2\Delta t}\, \Delta t ]^\dagger\; \cdots\; \exp[ \hat L_{N \Delta t}\, \Delta t ]^\dagger \right) ,
\end{equation}
where $\Delta t \equiv \tau / N$.
Upon choosing an orthonormal basis of $\Operators$ with $(\dim \Hilbert)^{-1/2}\, \mathbbm 1$ as its first element, the matrix representation of a single time-slice of the propagator \eqref{eq:toe} becomes
\begin{equation} \label{eq:slice}
	\exp[ \hat L_{k \Delta t}\, \Delta t ]^\dagger = \left[\! \begin{array}{cc}
			1 & \ast \\
			\vec 0 & \left(\, \exp[ \hat L_{k \Delta t}^\prime\, \Delta t ]^\dagger \,\right)
		\end{array} \!\right] ,
\end{equation}
where $\vec 0$ is a column vector with zero entries and $\hat L_t'$ is the restriction of $\hat L_t$ on $\mathbbm 1^\perp$.

We now apply the adjoint propagator $V_{\tau,0}^\dagger$ to an arbitrary observable $X_0 \in \Operators$.
Decomposing $X_t$ as $X_t \equiv x_t\, \mathbbm 1 + X'_t$ with $X'_t \in \mathbbm 1^\perp$ and using \eqref{eq:slice} yields
\begin{equation}
	X'_\tau = \lim_{N \to \infty} \left( \exp[ \hat L_{\Delta t}'\, \Delta t ]^\dagger\; \exp[ \hat L_{2\Delta t}'\, \Delta t ]^\dagger\; \cdots\; \exp[ \hat L_{N \Delta t}'\, \Delta t ]^\dagger \right) X'_0 .
\end{equation}
This expression makes it possible to derive an upper bound on $\norm{X'_\tau}_\infty$ in terms of $\norm{X'_0}_\infty$.
Here, the role of the subspace $\Operators'$ in the definition \eqref{eq:norm} of the $\infty$-norm is played by $\mathbbm 1^\perp$.
Recalling \eqref{eq:1norm}, we find that $\norm{ \exp[ \hat L_t'\, \Delta t ]\, }_1' \leq \ee^{-\lambda_t\, \Delta t}$ for some $\lambda_t > 0$, since $\hat L_t$ satisfies the conditions of Theorem~\ref{thm:spohn}.
Applying Lemma~\ref{lemma} to $\hat U = \exp[ \hat L_t'\, \Delta t ]^\dagger$ in every individual time-slice, we thus obtain
\begin{equation}
	\norm{X'_\tau}_\infty \leq \lim_{N \to \infty} \ee^{-\lambda_{\Delta t}\, \Delta t}\, \ee^{-\lambda_{2\Delta t}\, \Delta t}\, \cdots\, \ee^{-\lambda_{N \Delta t}\, \Delta t}\, \norm{X'_0}_\infty .
\end{equation}

The value of $\lambda_t$ can be found at each time $t$ by following the procedure described in \Sref{sec:no_driving}.
By assumption, $\hat L_t$ is continuous for $0 \leq t \leq \tau$.
Furthermore, we can assume without loss of generality that $\dim \mathcal A_t$ is constant throughout this time interval.
Therefore, we are free to choose the self-adjoint operators $F^\alpha_t$ appearing in the decomposition \eqref{eq:lindblad_expanded} as continuous functions of time.
By construction, the diagonal generator $\hat L_{\mathrm{d},t}$ and its largest non-zero eigenvalue, $-\lambda_t$, are then also continuous.
Therefore, $\Lambda \equiv \min_{0 \leq t \leq \tau} \lambda_t$ is strictly positive and we can conclude that
\begin{equation}
	\norm{X'_\tau}_\infty \leq \ee^{-\Lambda\tau}\, \norm{X'_0}_\infty .
\end{equation}
Hence, the $\infty$-norm of $X'$ strictly decreases under the action of the adjoint propagator $\hat V_{\tau,0}^\dagger$.

The propagator over the remaining part of the period, $\hat V_{T,\tau}^\dagger$, can be treated analogously.
Here, we use that $\norm{ \exp[ \hat L_{t}^\prime\, \Delta t ]\, }_1^\prime \leq 1$ for any Lindblad generator $\hat L_t$, i.e., the propagator $\hat V_{T,\tau}^\dagger$ can only decrease the $\infty$-norm of the $\mathbbm 1^\perp$-component.
For the time evolution over a full period, we thus obtain
\begin{equation}
	\norm{X'_T}_\infty \leq \e^{-\Lambda \tau} \norm{X'_0}_\infty ,
\end{equation}
where $X_T = \hat V_{\tau,0}^\dagger \hat V_{T,\tau}^\dagger\, X_0$.
It follows that $X_t = x_t \mathbbm 1$ at long times.

It remains to prove that $x_{t+T} = x_t$.
To this end, we compare the expressions
\begin{equation}
	X_t = \hat V_{t,0}^\dagger\, X_0
	\qquad \textrm{and} \qquad
	X_{t+T} = \hat V_{T,0}^\dagger\, \hat V_{t+T,T}^\dagger\, X_0 .
\end{equation}
The propagators are invariant under a global time shift $T$ due to the periodicity of $\hat L_t$, i.e., $\hat V_{t+T,T}^\dagger = \hat V_{t,0}^\dagger$.
Since $\mathbbm 1$ is an eigenvector of $\hat V_{T,0}^\dagger$, we obtain $X_{t+T} = \hat V_{T,0}^\dagger\, X_t = X_t$ at long times.
Thus, our proof of Theorem~\ref{thm:us} is complete.

\section{Concluding Remarks} \label{sec:conclusions}

In this letter, we have shown that a periodically driven open quantum system approaches a unique limit cycle if the dissipative terms of the GKSL equation mix all subspaces of the system Hilbert space during a finite fraction $\tau/T$ of each driving period.
In addition, our proof provides the lower limit $\Lambda\tau / T$ on the average rate of approach to the limit cycle, where the characteristic rate $\Lambda$ can be found from the Lindblad generator $\hat L_t$.
We note, however, that our proof of the existence of the limit cycle is not constructive.
In fact, the limit cycle can be determined explicitly only for specific systems or limiting regimes, see for example~\cite{FeldmannPhysRevE2004,BrandnerPhysRevE2016,KosloffEntropy2017,ScopaPhysRevA2018}; further characterizing the properties of these states on a general level appears to be impossible.

As one of its key applications, Theorem~\ref{thm:us} lays down the conditions under which a dissipative discrete time crystal can emerge in a Markovian environment.
To this end, the existence of a unique limit cycle with the same period $T$ as the driving must be avoided.
Hence, the set of Lindblad operators has to be either non-self-adjoint or reducible throughout the period.
This requirement ensures that the GKSL dynamics has either no or at least two distinct attractors in the space of system states.
Such a situation can be created, for example, by applying control protocols that unitarily transfer the system state from one basin of attraction to another at the end of a driving period, thus generating a system response with a period length greater than the applied cycle time, see~\cite{GongPhysRevLett2018, GambettaPhysRevLett2019} for further details.

A second major implication of Theorem~\ref{thm:us} arises from the observation that any classical Markov jump process can be described in terms of a GKSL equation.
Specifically, consider a probability vector $\vec P_t = (P^1_t, \dots, P^N_t)$ evolving according to the time-inhomogeneous Pauli master equation
\begin{equation} \label{eq:classical}
	\dot{\vec P}_t = \mathbb W_t\, \vec P_t ,
\end{equation}
where the transition rate matrix $\mathbb W_t$ contains the rates $(\mathbb W_t)_{ij} = \gamma_{ij,t}$ for $i \neq j$ and $(\mathbb W_t)_{ii} = -\sum_{j \neq i} \gamma_{ji,t}$.
For a diagonal density matrix $\rho_t = \sum_{i=1}^N P^i_t\, \ketbra i i$, the population dynamics of the GKSL equation
\begin{equation}
	\dot \rho_t = \sum_{i \neq j} \gamma_{ij,t}\, \Big( A_{ij} \rho_t A_{ij}^\dagger - \frac 1 2 A_{ij}^\dagger A_{ij} \rho_t - \frac 1 2 \rho_t A_{ij}^\dagger A_{ij} \Big)
\end{equation}
are equivalent to the master equation~\eqref{eq:classical} if $A_{ij} = \ketbra i j$.
Applying Theorem~\ref{thm:us}, we thus obtain the following corollary:
\begin{cor}
	Consider a system obeying the master equation \eqref{eq:classical} with periodic transition rate matrix, $\mathbb W_{t+T} = \mathbb W_t$.
	Assume there is a $\tau \in (0, T]$ so that $\mathbb W_t$ is continuous on $[0,\tau]$.
	For all $t \in [0, \tau]$, assume further that $(\mathbb W_t)_{ij} = 0$ if $(\mathbb W_t)_{ji} = 0$, and that $\mathbb W_t$ can not be made block-diagonal by permuting its rows and columns.
	Then there exists a unique limit cycle $\vec P_t^{\mathrm{cy}} = \vec P^{\mathrm{cy}}_{t+T}$ so that
	\begin{equation}
		\lim_{t \to \infty} \left( \vec P_t - \vec P_t^{\mathrm{cy}} \right) = 0
	\end{equation}
	for arbitrary initial conditions $\vec P_0$.
\end{cor}

While we have here focussed on periodically driven systems, our method can be applied also for non-periodic, e.g., quasi-periodic, driving.
To this end, we introduce the following generalization of the rate $\lambda_t$ to times $t$ where $\hat L_t$ does not satisfy the conditions of Theorem~\ref{thm:spohn}:
	if the span of all Lindblad operators at the time $t$ is not self-adjoint, we set $\lambda_t = 0$.
Otherwise, we can define the diagonal contribution $\hat L_{\mathrm{d},t}$ as in Section~\ref{sec:no_driving}, and $-\lambda_t \leq 0$ is the largest eigenvalue of $\hat L_{\mathrm{d},t}'$, i.e., of the restriction to the subspace of traceless operators.
The following corollary then follows along the lines of Theorem~\ref{thm:us}:

\begin{cor}
	For a system obeying the GKSL equation \eqref{eq:lindblad}, let $\lambda_t$ be defined as described above.
	If
	\begin{equation}
		\int_0^\infty \lambda_t \id t = \infty ,
	\end{equation}
	the system is relaxing, i.e., its behavior at long times is independent of the initial conditions.
\end{cor}

\begin{figure}
	\centering
	\includegraphics[scale=1]{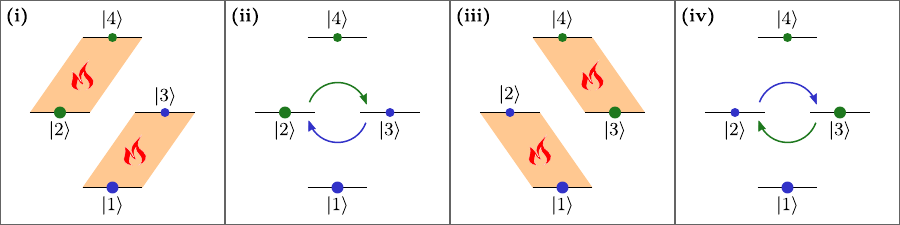}
	\caption[bla]{
		Periodically driven quantum four-level system without a unique limit cycle.
		The states $\ket 1$ through $\ket 4$ are indicated by horizontal lines in the diagrams; the filled circles on the lines indicate the level populations at the beginning of each step.
		We consider a cyclic four-step protocol:
		\begin{enumerate*}
		\item The populations of the pairs of states $(\ket 1, \ket 3)$ and $(\ket 2, \ket 4)$ thermalize separately.
			That is, the time evolution during this step follows \eqref{eq:lindblad} with constant Lindblad generator $\hat L_t = \hat L$.
			The generator contains four dissipative terms with the Lindblad operators $\ketbra 1 3$, $\ketbra 2 4$, and their adjoints.
		\item The populations of the states $\ket 2$ and $\ket 3$ are swapped unitarily.
		\item The populations of the pairs $(\ket 1, \ket 2)$ and $(\ket 3, \ket 4)$ of states thermalize separately.
		\item The populations of the states $\ket 2$ and $\ket 3$ are swapped unitarily.
		\end{enumerate*}
		The ratio between the populations marked in green and the populations marked in blue is determined by the initial conditions and remains unaltered over one period.
		Hence, the long time state cannot be unique.
	}
	\label{fig}
\end{figure}

Finally, it is a natural question to ask whether the conditions of Theorem~\ref{thm:us} could be weakened.
A weaker condition might, for example, only require that $\bigcup_{t \in [0,T]} \mathcal A_t$ is irreducible, i.e., that the dissipative dynamics connects all subspaces over the course of one driving period.
A simple counterexample, illustrated in \Fref{fig}, shows, however, that this condition is not sufficient for the limit cycle to be unique.
Hence, potential generalizations of Theorem~\ref{thm:us} would most likely require a closer analysis of the interplay between the unitary and the dissipative parts of the GKSL equation.
It remains a challenge for future investigations to settle the question whether such extensions can be formulated in terms of simple algebraic conditions.

\ack

We thank J.~P.~Garrahan for useful comments.
K.~B. acknowledges support from the Academy of Finland (Contract No.~296073).
This work was supported by the Academy of Finland (projects No.~308515 and 312299).
All authors are associated with the Centre for Quantum Engineering at Aalto University.

\newcommand{\newblock}{}

%

\end{document}


\articletype{Corrigendum}

\title{Corrigendum: Limit cycles in periodically driven open quantum systems (2019 J.\ Phys.\ A: Math.\ Theor.\ \href{https://doi.org/10.1088\%2F1751-8121\%2Fab435a}{\textbf{52} 43LT01)}}

\author{Paul Menczel$^{1,\ast}$\orcid{0000-0002-3968-0918} and Kay Brandner$^{2,3}$\orcid{0000-0002-4425-8252}}

\affil{$^1$RIKEN Center for Quantum Computing, RIKEN, Wakoshi, Saitama 351-0198, Japan}

\affil{$^2$School of Physics and Astronomy, University of Nottingham, Nottingham NG7 2RD, United Kingdom}

\affil{$^3$Centre for the Mathematics and Theoretical Physics of Quantum Non-Equilibrium Systems,
University of Nottingham, Nottingham NG7 2RD, United Kingdom}

\affil{$^\ast$Author to whom any correspondence should be addressed.}

\email{paul@menczel.net}


\begin{abstract}
We identify an essential error in our letter \cite{MenczelJPhysA2019}, providing explicit counterexamples.
We clarify that the derivations presented there are valid only for the special class of unital GKSL generators and outline the changes necessary to account for this restriction.
We also discuss similar errors appearing in the literature, and explain to what extent the results of the original letter remain valid on the basis of alternative proofs.
\end{abstract}

\section{Outline}

The main result of our letter \emph{Limit cycles in periodically driven open quantum systems}, which we henceforth refer to as the \emph{main text}, can be summarized as follows.
Consider an open quantum system with a finite-dimensional Hilbert space, whose state $\rho_t$ evolves according to the GKSL equation
\begin{equation} \label{eq:GKSL}
	\partial_t \rho_t = \hat{L}_t \rho_t
		\equiv \frac{1}{\ii \hbar} \comm{H_t}{\rho_t} + \sum_{\mu=1}^{M_t} \gamma^\mu_t \Bigl(
			A^\mu_t \rho_t A^{\mu\dagger}_t
				- \frac{1}{2} A^{\mu\dagger}_t A^\mu_t \rho_t
				- \frac{1}{2} \rho_t A^{\mu\dagger}_t A^\mu_t
		\Bigr) \,,
\end{equation}
where $\hbar$ denotes Planck's constant, $H_t$ is a self-adjoint operator and $M_t$ is the number of dissipation channels with corresponding Lindblad operators $A^\mu_t$ and rates $\gamma^\mu_t > 0$.
If the GKSL generator is periodic in time, $\hat{L}_{t+T} = \hat{L}_t$, the following theorem holds.

\begin{theorem}{2} \label{thm:2}
	If there is a $\tau \in (0,T]$ so that $\hat{L}_t$ is continuous on $[0,\tau]$, and $\sA_t = \lspan\bigl\{A^1_t,\dots, A^{M_t}_t\bigr\}$ is self-adjoint and irreducible for all $t\in[0,\tau]$, there exists a unique limit cycle $\rho^\cy_t = \rho^\cy_{t+T}$ so that
	\begin{equation}
		\lim_{t\to\infty} \bigl( \rho_t - \rho^\cy_t \bigr) = 0
	\end{equation}
	for arbitrary initial conditions $\rho_0$.
\end{theorem}

\noindent
Here, a set of operators $\sA$ is called self-adjoint if $A\in\sA$ implies $A^\dagger\in \sA$, and irreducible if $[X,A]=0$ for all $A\in\sA$ implies that $X$ is a multiple of the identity operator.

In Ref.~\cite{MenczelArXiv261029084Cond-Matstat-Mech2026}, it is shown that the above theorem is correct as stated.
However, the earlier proof in the main text contains an essential error rendering it invalid unless the Lindblad generator is also required to be unital, that is, to satisfy the condition $\hat{L}_t \mathbbm{1} = 0$, for all times;
	the limit cycle is then trivially given by $\rho^\cy_t = \mathbbm{1}/d$, where $d$ is the dimension of the system Hilbert space.

The purpose of this Corrigendum is twofold.
First, in Sec.~\ref{sec:flaw}, we explain the error in the main text, providing an explicit counterexample to a key inequality, which the proof in the main text relies on.
We further point out that similar errors appear in several places in the literature.
Second, in Sec.~\ref{sec:correction}, we explain how the proofs in the main text can be repaired when restricted to the special class of unital generators, for which its results and derivations remain valid.
At the end of Sec.~\ref{sec:correction}, we indicate which of the statements claimed in the main text have been proven elsewhere and therefore remain valid as stated.
For equations and theorems introduced in the main text, we keep the original numbers, using primes to denote corrected versions. 
Newly introduced equations are indicated by numbers with the prefix C.

\section{Errors} \label{sec:flaw}

\subsection*{Error in the Main Text}

The proof of Theorem~\ref{thm:2} in the main text builds on the earlier work \cite{SpohnLettMathPhys1977}, which seeks to establish the following result.

\begin{theorem}{1}[Spohn] \label{thm:1}
	For a system obeying the GKSL equation \eqref{eq:GKSL} with a time-independent Lindblad generator $\hat{L}$, if $\sA=\lspan\mathopen\{A^1, \dots, A^M\mathclose\}$ is self-adjoint and irreducible, then there exists a unique steady state $\rho^\mathrm{ss}$ so that
	\begin{equation}
		\lim_{t\to\infty} \bigl( \rho_t - \rho^\mathrm{ss} \bigr) = 0
	\end{equation}
	for arbitrary initial conditions $\rho_0$.
\end{theorem}

The main idea of the proof, which we adopted in the main text, is to divide the generator into two parts, $\hat{L} = \hat{L}_\mathrm{d} + \hat{L}_\mathrm{r}$, where both $\hat{L}_\mathrm{d}$ and $\hat{L}_\mathrm{r}$ have GKSL form and $\hat{L}_\mathrm{d}$ is self-adjoint with respect to the Hilbert-Schmidt inner product.
Furthermore, the kernel of $\hat{L}_\mathrm{d}$ is one-dimensional and contains only scalar multiples of the identity operator.
It then follows that the spectral radius of $\exp[\hat{L}'_\mathrm{d} t]$ is given by $\ee^{-\lambda t}$, where $\hat{L}'_\mathrm{d}$ is the restriction of $\hat{L}_\mathrm{d}$ to the subspace of traceless operators and $-\lambda < 0$ denotes the largest non-vanishing eigenvalue of $\hat{L}_\mathrm{d}$.
From these properties, the bounds
\begin{align}
	\norm[\big]{ \exp\bigl[ \hat L_\mathrm{d}' t \bigr] }_1' & \leq \ee^{-\lambda t} \tag{10} \label{eq:10}
\shortintertext{and}
	\norm[\big]{ \exp\bigl[ \hat L_\mathrm{r}' t \bigr] }'_1 & \leq 1 \tag{11} \label{eq:11}
\end{align}
are deduced, where $\norm{\,\cdot\,}'_1$ denotes the operator norm with respect to the restriction of the trace norm $\norm{\,\cdot\,}_1$ to the subspace of traceless operators.
That is, this norm is defined as
\begin{equation}
	\norm{\hat V}'_1 \equiv \max_X \bigl\{ \norm{\hat V X}_1 : \norm{X}_1 = 1 \text{ and } \tr[X] = 0 \bigr\} \,,
\end{equation}
where $\hat V$ is an arbitrary superoperator, i.e., a linear map on the operator space.
If the bounds \eqref{eq:10} and \eqref{eq:11} were both valid, they would, by means of the Lie-Trotter formula, imply that
\begin{equation}
	\norm[\big]{ \exp\bigl[ \hat L' t \bigr] }'_1 \leq \ee^{-\lambda t} \,, \tag{12} \label{eq:12}
\end{equation}
where $\hat{L}' = \hat{L}'_\mathrm{d}+\hat{L}'_\mathrm{r}$.

These steps underpin the proof of Theorem~\ref{thm:1} in Ref.~\cite{SpohnLettMathPhys1977} and that of Theorem~\ref{thm:2} in the main text.
They are also used in Ref.~\cite{SpohnRepMathPhys1976} to obtain a different statement.
Yet, the argument is incorrect, since the operator norm of an operator or superoperator can in general be larger than its spectral radius.
Both of the bounds \eqref{eq:10} and \eqref{eq:12} are therefore invalid, as the following counterexample demonstrates\footnote{Lasse H.\ Wolff, private communication (2026).}.

\subsection*{Counterexample}

Consider a four-dimensional Hilbert space with basis $\{\ket{i}\}_{i=1,\dots,4}$ and construct a time-independent GKSL generator $\hat{L}^1$ with zero Hamiltonian, $H=0$, and two self-adjoint Lindblad operators
\begin{equation}
	A^{1} = \frac{1}{\sqrt{6}} \bigl( a + a^\dagger \bigr) \qquad\text{and}\qquad A^{2} = \frac{\ii}{\sqrt{6}} \bigl( a - a^\dagger \bigr) \,.
\end{equation}
The corresponding rates are set to $\gamma^{1}=\gamma^{2} = 6$, and $a$ denotes the ladder operator $a = \ketbra 1 2 + \ketbra 2 3 + \ketbra 3 4$.

It is straightforward to verify that the only operators that commute with both $A^{1}$ and $A^{2}$ are scalar multiples of the identity.
The generator $\hat{L}^1$ therefore satisfies the conditions of Theorem~\ref{thm:1}.
Furthermore, since $A^{1}$ and $A^{2}$ are orthonormal with respect to the Hilbert-Schmidt inner product, the generator $\hat L^1$ already has the form of Eq.~(8) of the main text with $F^\alpha = A^\alpha$ ($\alpha = 1,2$) and $B_{\alpha\beta} = 6\, \delta_{\alpha\beta}$.
Following Eq.~(9) in the main text, we therefore have $\hat L_\mathrm{d}^1 = \hat L_\mathrm{r}^1 = \hat L^1 / 2$.
The propagator $\exp[ \hat L^1 t ]$ can be computed in a suitable matrix representation.
For the traceless operator $X^1 = ( \ketbra 1 1 - \ketbra 4 4 ) / 2$, which satisfies $\norm{X^1}_1 = 1$, we obtain
\begin{align}
	\exp\bigl[\hat L^{1\prime}_\mathrm{d} t\bigr] X^1 &= \exp\bigl[\hat L^{1\prime} t / 2\bigr] X^1 \nonumber\\
		&= \frac{\ee^{-2t}}{2} \bigl(2\cosh\bigl[ \sqrt{2} t \bigr] + \sqrt{2}\sinh\bigl[\sqrt{2} t \bigr] \bigr) X^1
			+ \frac{\ee^{-2t}}{2\sqrt{2}} \sinh\bigl[\sqrt{2} t \bigr]\, \mathopen( \ketbra 2 2 - \ketbra 3 3 \mathclose) \,,
\end{align}
and therefore
\begin{align}
	\norm[\big]{\exp\bigl[\hat L^{1\prime}_\mathrm{d} t\bigr] X^1}_1
		&= \norm[\big]{\exp\bigl[\hat L^{1\prime} t / 2\bigr] X^1}_1 \nonumber\\
		&= \ee^{-2t} \bigl( \cosh\bigl[\sqrt{2} t \bigr] + \sqrt{2}\sinh\bigl[\sqrt{2} t \bigr] \bigr) \nonumber\\
		&= 1 - t^2 + \mathcal O(t^3) \, .
\end{align}
The absence of a linear term in $t$ shows that this expression cannot be bounded by $e^{-\lambda t}$ for any $\lambda>0$.
Hence, the same conclusion holds for
\begin{align}
	\norm[\big]{\exp[\hat{L}^{1\prime}_\mathrm{d}t]}^\prime_1
		&= \max_X \bigl\{ \norm[\big]{\exp[\hat{L}^{1\prime}_\mathrm{d}t]X}_1 : \norm{X}_1 = 1 \text{ and } \tr[X] = 0 \bigr\} \nonumber\\
		&\geq \norm[\big]{\exp[\hat{L}^{1\prime}_\mathrm{d}t]X^1}_1 \,,
\end{align}
and equally for $\norm[\big]{\exp[\hat{L}^{1\prime}t]}'_1 = \norm[\big]{\exp[2\hat{L}^{1\prime}_{\mathrm{d}}t]}'_1$.
It follows that neither Eq.~\eqref{eq:10} nor Eq.~\eqref{eq:12} can hold for the generator $\hat{L}^1$.

\subsection*{Other Norms}

One may wonder whether Eqs.~\eqref{eq:10} and \eqref{eq:11}, and by extension Eq.~\eqref{eq:12}, would be valid in a different norm, in particular, since the applied norm is not unambiguously specified in Refs.~\cite{SpohnRepMathPhys1976,SpohnLettMathPhys1977}.
For example, since $\hat{L}'_\mathrm{d}$ is self-adjoint with respect to the Hilbert-Schmidt inner product by construction, the bound \eqref{eq:10} becomes an identity for the Hilbert-Schmidt norm $\norm{\,\cdot\,}_2$.
However, GKSL generators are generally not contractive with respect to this norm \cite{Perez-GarciaJMathPhys2006}.
Therefore, Eq.~\eqref{eq:11} does not hold when the $1$-norm is replaced by the $2$-norm, and Eq.~\eqref{eq:12} remains invalid.

In fact, the following counterexample shows that Eq.~\eqref{eq:12} cannot hold for any Schatten $p$-norm with $p \in (1, \infty)$.
We recall that, for any $p\in [1,\infty)$, the $p$-norm of a self-adjoint operator $X$ on a Hilbert space with $d$ dimensions is given by
\begin{equation}
	\norm{X}_p = \bigl( \abs{x_1}^p + \cdots + \abs{x_d}^p \bigr)^{1/p} \,,
\end{equation}
where $x_1, \dots, x_d$ are the eigenvalues of $X$.
We now take $d=3$ and construct a GKSL generator $\hat{L}^2$ with $H = 0$ and a single Lindblad operator $A = \ketbra 21$ with rate $\gamma=1$.
For the traceless, self-adjoint operator $X^2 = \ketbra 1 1 + \ketbra 2 2 - 2 \ketbra 3 3$, we obtain
\begin{align}
	\norm[\big]{\exp\bigl[\hat L^{2\prime} t\bigr] X^2}_p^p
		&= \frac{1}{2+2^p}\, \bigl[2^p+\ee^{-pt} + \bigl(2-\ee^{-t}\bigr)^p \bigr]\, \norm{X^2}_p^p \nonumber\\
		&= \left[ 1 + \frac{p(p-1)}{2+2^p} t^2 \right]\, \norm{X^2}_p^p + \mathcal O(t^3) \,.
\end{align}
Since the coefficient of $\norm{X^2}_p^p$ in the second line is a strictly increasing function of $t$ for any $p\in (1,\infty)$, it follows that $\norm{\exp[\hat L^{2\prime} t]}'_p > 1$ for sufficiently small $t > 0$.
We note that the generator $\hat{L}^2$ does not satisfy the conditions of Theorem~\ref{thm:1}.
However, due to continuity of the norm, it is possible to add additional Lindblad operators with sufficiently small rates to satisfy these conditions without invalidating the above conclusion.

\subsection*{Spectral Properties of Sums of Operators}

It has been suggested in Ref.~\cite{Rivas2012} to correct the proof of Theorem~\ref{thm:1} through an argument based on the spectral properties of $\hat{L}'_\mathrm{d}$ and $\hat{L}'_\mathrm{r}$, which circumvents the use of norm inequalities.
Specifically, it is argued there that since the spectrum $\sigma(\hat{L}'_\mathrm{d})$ of the restricted self-adjoint generator $\hat{L}'_\mathrm{d}$ is strictly negative, and the spectrum $\sigma(\hat{L}'_\mathrm{r})$ of the restricted remainder $\hat{L}'_\mathrm{r}$ has non-positive real part, the spectrum of their sum must have strictly negative real part.
This statement is indeed true if $\hat{L}_\mathrm{d}$ and $\hat{L}_\mathrm{r}$ are GKSL generators and can be verified \emph{a posteriori} using Theorem~1, which is in fact correct.
However, it is not proven in Ref.~\cite{Rivas2012}.
Instead, the derivation presented there uses the more general claim that, if $\hat{V}_\mathrm{a}$ and $\hat{V}_\mathrm{b}$ are two superoperators such that $\hat{V}_\mathrm{a}$ is self-adjoint, $\sigma(\hat{V}_\mathrm{a})<0$ and $\operatorname{Re}[\sigma(\hat{V}_\mathrm{b})]\leq 0$, then $\operatorname{Re}[\sigma(\hat{V}_\mathrm{a} + \hat{V}_\mathrm{b})]<0$.
This conclusion is invalid.
For completeness, we provide a counterexample below.
Consider the two superoperators
\begin{align}
	\hat{V}^3_\mathrm{a}X &= - \frac{5}{4}X -\frac{1}{4} \sum_{\alpha=x,y,z} \sigma_\alpha X \sigma_\alpha \,, \\
	\hat{V}^3_\mathrm{b}X & = -\frac{1}{4}X -X\sigma_x - \sigma_x X
		- \frac{1}{4}\sigma_x X\sigma_x
		- \frac{9}{4}\sigma_y X\sigma_y
		- \frac{9}{4}\sigma_z X\sigma_z
		- \frac{5\ii}{2}\sigma_y X\sigma_z
		+ \frac{5\ii}{2}\sigma_zX\sigma_y \,,
\end{align}
where $\sigma_{x,y,z}$ are the usual Pauli matrices.
Here, $\hat{V}^3_\mathrm{a}$ is self-adjoint and $\sigma(\hat{V}_\mathrm{a}^3) = \{-2,-1,-1,-1\}<0$, while $\operatorname{Re}[\sigma(\hat{V}^3_\mathrm{b})]=\{-1/2,-1/2,0,0\}\leq 0$.
However, $(\hat{V}^3_\mathrm{a}+\hat{V}_\mathrm{b}^3)(\mathbbm{1}-\sigma_x)=0$.
Hence, the spectrum of the sum of $\hat{V}^3_\mathrm{a}$ and $\hat{V}_\mathrm{b}^3$ cannot have strictly negative real part.

\section{Corrections} \label{sec:correction}

For unital GKSL generators, which satisfy the condition $\hat{L}_t \mathbbm{1}=0$, the derivations of the main text remain valid with minor changes.
Under this condition, any unique long-time state is proportional to the identity operator, regardless of whether or not the generator is time-dependent.
In the following, we correct the statements of the main text one by one to include the additional condition of unitality.
We briefly discuss how their proofs must be adapted.

\subsection*{Changes to Section 1}

Theorems 1 and 2 must include the condition of unitality.
The modified theorems are stated below.

\begin{theorem}{1'} \label{thm:1new}
	For a system obeying the GKSL equation (1) with a time-independent, unital Lindblad generator $\hat L$, let $\mathcal A \equiv \lspan\mathopen\{A^1 \dots A^M\mathclose\}$ be the linear span of all Lindblad operators.
	If $\mathcal A$ is self-adjoint and irreducible, then there exists a unique state $\rho^\mathrm{ss}$ so that
	\begin{equation}
		\lim_{t \to \infty} \bigl( \rho_t - \rho^\mathrm{ss} \bigr) = 0
	\end{equation}
	for arbitrary initial conditions $\rho_0$.
	Here, $\mathcal A$ is self-adjoint if $X^\dagger \in \mathcal A$ for all $X \in \mathcal A$ and irreducible if $\comm X Y = 0$ for all $X \in \mathcal A$ implies that $Y$ is a multiple of the identity.
\end{theorem}

\begin{theorem}{2'} \label{thm:2new}
	Consider a system obeying the GKSL equation (1) with periodic Lindblad generator, $\hat L_{t + T} = \hat L_t$.
	Assume that $\hat L_t$ is unital at all times, and that there is a $\tau \in (0, T]$ so that $\hat L_t$ is continuous on $[0, \tau]$, and $\mathcal A_t \equiv \lspan\bigl\{ A^1_t \dots A^{M_t}_t \bigr\}$ is self-adjoint and irreducible for all $t \in [0, \tau]$.
	Then there exists a unique limit cycle $\rho^\mathrm{cy}_t = \rho^\mathrm{cy}_{t+T}$ so that
	\begin{equation}
		\lim_{t \to \infty} \bigl( \rho_t - \rho^\mathrm{cy}_t \bigr) = 0
	\end{equation}
	for arbitrary initial conditions $\rho_0$.
\end{theorem}

\subsection*{Changes to Section 2}

The trace norm $\norm{\,\cdot\,}_1$ needs to be replaced by the Hilbert-Schmidt norm $\norm{\,\cdot\,}_2$ on the space $\mathcal{S}$ throughout.
The operator norm of a superoperator $\hat V$ with respect to this norm is defined as
\begin{equation}
	\norm{\hat V}_2 \equiv \max_X \bigl\{ \norm{\hat V X}_2 : X \in \mathcal S \text{ and } \norm{X}_2 = 1 \bigr\} \,.
\end{equation}

\subsection*{Changes to Section 3}

The proof of Theorem~\ref{thm:1new} proceeds unchanged until equation~\eqref{eq:10}, which must be replaced by
\begin{equation} \label{eq:10new}
	\norm[\big]{ \exp\bigl[ \hat L_\mathrm{d}' t \bigr] }_2' = \ee^{-\lambda t} \,. \tag{10'}
\end{equation}
As mentioned above, equality here stems from the fact that $\hat L_\mathrm{d}'$ is self-adjoint with respect to the Hilbert-Schmidt inner product, which is compatible with the $2$-norm.

The argument continues as follows.
By construction, the remainder $\hat L_\mathrm{r}$ is still the generator of a quantum dynamical semigroup.
Since $\hat L$ is unital by assumption, and $\hat L_\mathrm{d}$ is unital by construction, $\hat L_\mathrm{r}$ is unital as well.
Recall that GKSL generators are contractive with respect to the Hilbert-Schmidt norm if and only if they are unital \cite{Perez-GarciaJMathPhys2006}\footnote{
	Here, contractivity of a GKSL generator $\hat L$ with respect to the 2-norm means that $\norm{\exp[\hat L t]}_2 = 1$.
	Further necessary and sufficient conditions for the weaker requirement $\norm{\exp[\hat L' t]}_2' \leq 1$ are studied in Ref.~\cite{WangPhysRevA2009}.
}.
Therefore, $\norm{\exp[\hat L_\mathrm{r} t]}_2 = 1$.
Since the image of $\hat L_\mathrm{r}$ is contained in $\mathbbm 1^\perp$, we find that
\begin{equation} \label{eq:11new}
	\norm[\big]{\exp[\hat L_\mathrm{r}' t]}_2' \leq \norm[\big]{\exp[\hat L_\mathrm{r} t]}_2 = 1 \,. \tag{11'}
\end{equation}

From here, the argument continues as stated in the main text.
Upon replacing the $1$-norm by the $2$-norm in Eq.~\eqref{eq:12}, it becomes
\begin{equation}
	\norm[\big]{\exp[\hat L' t]}_2' \leq \ee^{-\lambda t} \,. \tag{12'}
\end{equation}

\subsection*{Changes to Section 4}

The $1$-norm must be replaced by the $2$-norm in the definition of $\norm{X}_\infty$ for $X \in \mathcal S'$, where $\mathcal S'$ is a subspace of $\mathcal S$;
Eq.~(14) then becomes
\begin{equation}
	\norm{X}_\infty \equiv \max_Y \bigl\{ \abs[\big]{ \langle Y, X \rangle } : Y \in \mathcal S' \text{ and } \norm{Y}_2 = 1 \bigr\} \,. \tag{14'}
\end{equation}
For convenience, we keep the notation $\norm{X}_\infty$, although this norm does not coincide with the operator norm for $\mathcal S' = \mathcal S$.
Lemma~3 must be changed as follows.

\begin{lemma}{3'} \label{thm:3new}
	Let $\hat U: \mathcal S' \to \mathcal S'$ be a superoperator on the subspace $\mathcal S'$.
	Then $\norm{\hat U}_\infty$, the operator norm of $\hat U$ with respect to $\norm{\,\cdot\,}_\infty$, satisfies
	\begin{equation}
		\norm[\big]{\hat U}_\infty \leq \norm[\big]{\hat U^\dagger}_2' \,, \tag{15'}
	\end{equation}
	where $\norm{\,\cdot\,}_2'$ is the operator norm with respect to the Hilbert-Schmidt norm on $\mathcal S'$.
\end{lemma}

Both Lemma~\ref{thm:3new} and Theorem~\ref{thm:2new} can be proven along the lines of the main text, where $\norm{\,\cdot\,}_1$ and $\norm{\,\cdot\,}_1'$ must be replaced by $\norm{\,\cdot\,}_2$ and $\norm{\,\cdot\,}_2'$ throughout.
Note that the assumption of unitality throughout the period ensures that $\norm{\exp[\hat L_t' \Delta t]}_2' \leq 1$ also holds outside of the time interval $[0, \tau]$.

\subsection*{Changes to Section 5}

Corollaries 4 and 5 must be corrected as follows.

\begin{corr}{4'}
	Consider a system obeying the master equation (24) with periodic transition rate matrix, $\mathbb W_{t+T} = \mathbb W_t$.
	Assume that the rows and columns of $\mathbb W_t$ sum to zero at all times.
	Further assume there is a $\tau \in (0, T]$ so that $\mathbb W_t$ is continuous on $[0, \tau]$.
	For all $t \in [0, \tau]$, assume further that $\bigl( \mathbb W_t \bigr)_{ij} = 0$ if $\bigl( \mathbb W_t \bigr)_{ji} = 0$, and that $\mathbb W_t$ can not be made block-diagonal by permuting its rows and columns.
	Then there exists a unique limit cycle $\mathbf P^\mathrm{cy}_t = \mathbf P^\mathrm{cy}_{t+T}$ so that
	\begin{equation}
		\lim_{t \to \infty} \bigl( \mathbf P_t - \mathbf P^\mathrm{cy}_t \bigr) = 0\tag{26'}
	\end{equation}
	for arbitrary initial conditions $\mathbf P_0$.
\end{corr}

\begin{corr}{5'}
	For a system obeying the GKSL equation (1) with unital generator $\hat L_t$, let $\lambda_t$ be defined as described above [see main text].
	If
	\begin{equation}
		\int_0^\infty \lambda_t\, \dd t = \infty \,, \tag{27'}
	\end{equation}
	the system is relaxing, i.e., its behavior at long times is independent of the initial conditions.
\end{corr}

\subsection*{Validity of Original Statements}

\textbf{Theorem~1} is correct as stated in the main text, and was proven as Theorem~3.2 in the more general setting of Ref.~\cite{FrigerioCommunMathPhys1978}.
We note that, in contrast to Ref.~\cite{FrigerioCommunMathPhys1978}, the main text and Ref.~\cite{SpohnLettMathPhys1977} focus on finite-dimensional Hilbert spaces, where every dynamical semigroup has a stationary state and every state is normal.
\textbf{Theorem~2} and \textbf{Corollary~4} are correct as stated in the main text; they are special cases of Corollaries 7 and 8 in Ref.~\cite{MenczelArXiv261029084Cond-Matstat-Mech2026}, respectively.
\textbf{Corollary~5} holds in the given form only for unital generators.
A related but different criterion for complete relaxation under non-periodically driven GKSL dynamics is provided in Ref.~\cite{MenczelArXiv261029084Cond-Matstat-Mech2026} as Theorem~9.

\ack{The authors thank Lasse H.\ Wolff for pointing out the essential error in Ref.~\cite{MenczelJPhysA2019}, and for providing the first counterexample discussed in Sec.~\ref{sec:flaw}.}

\funding{This work was supported by the Japan Science and Technology Agency (JST) via the CREST Quantum Frontiers Program [grant number JPMJCR24I2]; the Moonshot R\&D Program [grant number JPMJMS256E]; the Medical Research Council [grant numbers MR/S034714/1, MR/Y003845/1]; and the Engineering and Physical Sciences Research Council [grant number EP/V031201/1].}

\data{No new data were created or analysed in this study.}